\documentclass[12pt,american]{article}\usepackage[]{graphicx}\usepackage[]{xcolor}
\makeatletter
\def\maxwidth{ %
  \ifdim\Gin@nat@width>\linewidth
    \linewidth
  \else
    \Gin@nat@width
  \fi
}
\makeatother

\definecolor{fgcolor}{rgb}{0.345, 0.345, 0.345}

\usepackage{framed}
\makeatletter
 {\par\unskip\endMakeFramed%
 \at@end@of@kframe}
\makeatother

\definecolor{shadecolor}{rgb}{.97, .97, .97}
\definecolor{messagecolor}{rgb}{0, 0, 0}
\definecolor{warningcolor}{rgb}{1, 0, 1}
\definecolor{errorcolor}{rgb}{1, 0, 0}
\newenvironment{knitrout}{}{} 

\usepackage{alltt}
\PassOptionsToPackage{natbib=true}{biblatex}
\usepackage[T1]{fontenc}
\usepackage[utf8]{inputenc}
\usepackage{color}
\usepackage{babel}
\usepackage{cprotect}
\usepackage{varioref}
\usepackage{mathrsfs}
\usepackage{mathtools}
\usepackage{amsmath}
\usepackage{amsthm}
\usepackage{amssymb}
\usepackage{geometry}
\usepackage{setspace}
\usepackage[pdftex,pdfusetitle,
 bookmarks=true,bookmarksnumbered=false,bookmarksopen=false,
 breaklinks=false,pdfborder={0 0 0},pdfborderstyle={},backref=false,colorlinks=true]
 {hyperref}
\hypersetup{
 linkcolor = blue,urlcolor  = blue,citecolor = blue,anchorcolor = blue}

\makeatletter
\theoremstyle{definition}
\newtheorem{defn}{\protect\definitionname}
\theoremstyle{plain}
\newtheorem{lyxalgorithm}{\protect\algorithmname}
\theoremstyle{plain}
\newtheorem{assumption}{\protect\assumptionname}
\theoremstyle{plain}
\newtheorem{thm}{\protect\theoremname}
\theoremstyle{plain}
\newtheorem*{thm*}{\protect\theoremname}

\usepackage{arydshln}

\@ifpackageloaded{fixme}{\fxsetup{draft, targetlayout=changebar, targetface=\bfseries\itshape}}

\newcommand{\oldCB}{standard CB}
\newcommand{\OldCB}{Standard CB}

\ifdefined\showcaptionsetup
 \PassOptionsToPackage{caption=false}{subfig}
\fi
\usepackage{subfig}
\makeatother

\usepackage[style=authoryear,maxcitenames=2, maxbibnames=99]{biblatex}
\providecommand{\algorithmname}{Algorithm}
\providecommand{\assumptionname}{Assumption}
\providecommand{\definitionname}{Definition}
\providecommand{\theoremname}{Theorem}

\IfFileExists{upquote.sty}{\usepackage{upquote}}{}
\begin{document}
\title{Flatness-Robust Critical Bandwidth\thanks{I thank Chunrong Ai, Zheng Fang, and Douglas Turner for excellent
comments. I report there is no competing interest to declare.}}
\author{Scott Kostyshak\thanks{Assistant Professor, University of Florida, Dept.\ of Economics.
Email:~skostyshak@ufl.edu.}~\thanks{Visiting Professor, Universitat Pompeu Fabra, Dept.\ of Economics
and Business. Email:~scott.kostyshak@upf.edu.}}
\maketitle
\begin{abstract}
\noindent\setstretch{1.35}Critical bandwidth (CB) is used to test
the multimodality of densities and regression functions, as well as
for clustering methods. CB tests are known to be inconsistent if the
function of interest is constant (``flat'') over even a small interval,
and to suffer from low power and incorrect size in finite samples
if the function has a relatively small derivative over an interval.
This paper proposes a solution, flatness-robust CB (FRCB), that exploits
the novel observation that the inconsistency manifests only from regions
consistent with the null hypothesis, and thus identifying and excluding
them does not alter the null or alternative sets. I provide sufficient
conditions for consistency of FRCB, and simulations of a test of regression
monotonicity demonstrate the finite-sample properties of FRCB compared
with CB for various regression functions. Surprisingly, FRCB performs
better than CB in some cases where there are no flat regions, which
can be explained by FRCB essentially giving more importance to parts
of the function where there are larger violations of the null hypothesis.
I illustrate the usefulness of FRCB with an empirical analysis of
the monotonicity of the conditional mean function of radiocarbon age
with respect to calendar age.
\end{abstract}

\newpage{}

\global\long\def\hstatCB{h^{CB}}%

\global\long\def\hstatFRCB{h^{FRCB}}%

\global\long\def\convergesInDistribution{\overset{d}{\rightarrow}}%

\global\long\def\convergesInProbability{\overset{p}{\rightarrow}}%

\global\long\def\P{\mathrm{P}}%

\global\long\def\E{\mathrm{E}}%

\section{Introduction}

Critical bandwidth (CB), introduced by \citet{silverman_unimodal},
is used to test the multimodality\footnote{When discussing the number of modes (equivalently, ``peaks''), I
am referring to the weak concept. For example, when discussing monotonicity
I am referring to weak monotonicity; and I refer to a function $f$
as unimodal if it is weakly unimodal, i.e., there exists a value $m$
for which $f$ is monotonically increasing for $x\le m$ and monotonically
decreasing for $x\ge m$.} of densities and regression functions, to detect mixture distributions,
and as a component of clustering methods. CB is discussed in monographs
on the bootstrap as an innovative way to enforce the null hypothesis,\footnote{See, e.g., Section 16.5, ``Testing multimodality of a population,''
of \citet{efron_bootstrap}; \citet[p. 153]{hallbootstrap}; and \citet[p. 189]{davison1997bootstrap}.} and has applications across many fields. Examples of the diverse
applications of CB include identifying the number of growth spurts
in height \citep{harezlak_heckman}, exploring the multimodality of
labor productivity \citep{henderson2008_parmeter}, identifying determinants
of the U-shape of life satisfaction over the life cycle \citep{kostyshak_2017},
and testing the ecological niche separation of species with similar
dietary requirements \citep{cumming2017_niche_separation}.\footnote{Although CB can be applied to multi-dimensional functions of interest,
in this paper I focus on univariate functions because they are the
most common applications of CB.}

Despite the variety of null hypotheses that CB tests are used for,
the robustness of results based on CB tests has been limited because
they are generally inconsistent if the true function of interest is
constant (``flat'') over an interval. With finite sample sizes,
this flatness problem can arise in additional situations, even when
the true regression function or density does not contain a perfectly
flat region: If the derivative is close to zero in absolute value,
the CB test can suffer from low power and incorrect size, as shown
in the simulations in Section~\ref{sec:Simulations}. Although flatness
exclusion is a reasonable assumption in some situations, many real-world
relationships have regions over which the function of interest is
constant or has a small derivative. For example, the height of humans
is increasing in the early ages, flattens out, and then eventually
decreases at older ages.

Flatness exclusion was assumed in the original proof of CB consistency
\citep{silverman1983},\footnote{\citet{silverman1983} assumes bounded support of the density, and
that the derivative has ``no multiple zeros.''} and has been noted by \citet{mammen1992}, \citet{cheng1998},
and \citet{hallHeckman_mono_dip}.\footnote{The widespread recognition of the inconsistency is reflected by the
multiple terms used to refer to it: the ``spurious mode problem,''
the ``flatness problem,'' and the ``boundary problem.''} Further, new techniques that propose improvements to \oldCB{} also
assume flatness exclusion (e.g., \citet{ameijeiras2016_multimode}).
With no solution available, the flatness problem has motivated the
development of non-CB methods that do not suffer from the same inconsistency
(e.g., \citet{cheng1998}, \citet{hallHeckman_mono_dip}, and \citet{gijbels2000}).
Such methods are solutions for specific situations in which CB has
been applied (e.g., regression monotonicity), but no solution has
been proposed that applies to all situations for which the flexible
CB framework can be used.

This paper proposes a flatness-robust critical bandwidth (FRCB) test
that is valid without the assumption of flatness exclusion. Since
flat regions do not contradict the null hypothesis about the number
of peaks or valleys of the function of interest, a consistent test
can be constructed by identifying and excluding such regions. Formally,
consider a parameter space of functions $\Theta$. Let $\hat{f}_{f}$
be a semi-parametric estimator of $f\in\Theta$, where the subscript
emphasizes that the distribution of the estimator depends on the true
parameter $f$. Suppose that $\hat{f}_{f}$ is consistent for all
$f\in\Theta$. That is, suppose that for the domain of interest $\mathcal{X}\subset\mathbb{R}$,
$\sup_{x\in\mathcal{X}}\left|\hat{f}_{f}(x)-f(x)\right|\overset{p}{\rightarrow}0$.
Let the pair $(\Theta^{N},\Theta^{A})$ partition $\Theta$ into the
sets corresponding to the null and alternative hypotheses. Consider
the subspace $\Theta^{F}\subset\Theta$ of functions with flat regions
that has non-null intersections with both $\Theta^{N}$ and $\Theta^{A}$.\footnote{For examples of elements in the null and alternative sets with flat
or near-flat regions, see Figure~\vref{reg_fns_graph}.} \OldCB{} techniques exclude $\Theta^{F}$ from the parameter space
by assumption. That is, even if for all $f\in\Theta^{F}$, $\hat{f}_{f}$
converges in probability to $f$, \oldCB{} tests are still not consistent
for this parameter subspace. In this paper, I allow $\Theta^{F}$
to be part of the parameter space and provide a data-driven transformation
$\mathbf{T}$ such that for all $f\in\Theta^{N}$, $\mathbf{T}\hat{f}_{f}\overset{p}{\rightarrow}\tilde{f}_{f}$
for some $\tilde{f}_{f}\in\Theta^{N}\setminus\Theta^{F}$; and (for
the same $\mathbf{T}$) for all $f\in\Theta^{A}$, $\mathbf{T}\hat{f}_{f}\overset{p}{\rightarrow}\tilde{f}_{f}$
for some $\tilde{f}_{f}\in\Theta^{A}\setminus\Theta^{F}$. In other
words, for any element in $\Theta$, after the transformation \emph{standard}
CB tests are asymptotically valid, since the probability limit of
$\mathbf{T}\hat{f}_{f}$ is not in $\Theta^{F}$. the transformation
mechanism thus exploits that it is unnecessary, and in fact undesired
(because of the flatness problem), that $\mathbf{T}\hat{f}_{f}\overset{p}{\rightarrow}f_{f}$
for the class of multimodal hypothesis tests. 

The transformation discussed above, $\mathbf{T}$, relies on a uniform
confidence band for $f'$ over $\mathcal{X}$ that essentially filters
out potentially flat regions as a first step before applying \oldCB{}.
Asymptotically, as the confidence band shrinks, non-flat regions are
included while flat regions are not. Recent work in statistics has
developed results for confidence bands for $f'$ in different settings.
For non-parametric regression, \citet{belloni2015_uniform} provides
results for uniform confidence bands of linear functionals (which
include derivatives) of the regression function. Similarly, \citet{chen2018optimal_uniform}
provides results for non-parametric instrumental variables regression.
Any uniform confidence band can be improved by applying the post-estimation
procedures in \citet{chen2021shape}. For example, a researcher may
use FRCB to test monotonicity of a regression function, and could
gain efficiency by imposing convexity. FRCB thus builds on the long-standing
CB literature, and the more recent literature on uniform confidence
bands of derivatives, to provide a test that removes a previously-needed
assumption.

The rest of the paper is organized as follows. Section~\ref{sec:Critical-Bandwidth}
reviews the \oldCB{} test and the flatness problem, which causes
inconsistency, low power, and incorrect size. Section~\ref{sec:FRCB}
introduces the FRCB test, which forms the core of the paper, and provides
sufficient conditions for consistency of the test. Section~\ref{sec:Simulations}
compares the performance of FRCB to CB in simulations of a test of
regression monotonicity. Section~\ref{sec:Application} illustrates
the usefulness of FRCB with an empirical analysis of the monotonicity
of the conditional mean function of radiocarbon age with respect to
calendar age. Section~\ref{sec:Conclusion} concludes. Proofs of
the theorems are in the Appendix.

\section{Critical Bandwidth\label{sec:Critical-Bandwidth}}

In this section, I define the class of CB test statistics and give
specific examples that fit in the framework. Let $D$ be a random
matrix of data with support $\mathscr{D}$, $\mathcal{F}$ the parameter
space, $f\in\mathcal{F}$ the function of interest, $G\subseteq\mathbb{R}$
the grid over which $f$ is estimated,\footnote{Every CB implementation uses a grid to check that a non-parametric
estimate is of a certain shape. The grid can be as dense as desired.
In the simulations in Section~\ref{sec:Simulations}, a grid of 100
and a grid of 500 yield the same results for the hypothesis test.} and $\mathsf{H}$ the space of smoothing parameters (e.g., the space
of valid bandwidths in kernel regression). Let $\hat{f}:\mathscr{D}\times\mathsf{H}\times\mathbb{R}\rightarrow\mathcal{F}$
be the semi-parametric estimator that maps data and smoothing parameters
to the parameter space, evaluated on a grid. This setup covers many
situations in applied statistics. $D$ often contains the explanatory
variable of interest, covariates, and the outcome variable. For example,
even though we assume for simplicity that $f$ is univariate, often
covariates are controlled for additively (e.g., in partial linear
models or generalized additive models).

The class of null hypotheses for which CB applies can be characterized
by a property of the underlying estimator for functions in a null
set. Let $\mathcal{H}_{0}\subseteq\mathcal{F}$, $D\subseteq\mathscr{D}$,
$\{h,h'\}\subseteq\mathsf{H}$, and suppose the following property
holds:\footnote{For cases in which this property is not satisfied exactly, simulations
suggest that contradictions of the property rarely occur. For example,
\citet{bowman_monotonicity} find that ``out of the total of 90,000
simulations only two cases were discovered where the estimated regression
curve was monotonic at one bandwidth and nonmonotonic at a higher
one. This behavior is therefore extremely rare and its effect on the
test procedure will be negligible.''}
\[
\hat{f}(D,h,G)\in\mathcal{H}_{0},\,\,h'>h\implies\hat{f}(D,h',G)\in\mathcal{H}_{0}.
\]
For example, if $\mathcal{H}_{0}$ is the class of monotone functions,
satisfying this property requires that if a smoothing parameter value
yields an estimate that is monotone, increasing the smoothing parameter
from that value also yields an estimate that is monotone. The CB test
statistic is then defined as
\[
\hstatCB(D,G)\coloneqq\min\left\{ h\in\mathsf{H}\mid\hat{f}(D,h,G)\in\mathcal{H}_{0}\right\} .
\]
The null hypothesis is rejected if $\hstatCB(D,G)$ is too large,
which occurs when only a large smoothing parameter can force the estimator
into being consistent with the null hypothesis. Critical values are
determined from a bootstrap.\footnote{To create the bootstrapped data sets, a sample (with replacement)
is taken from the residuals and added to $\hat{f}(D,\hstatCB,G)$.
For more details, see \citet{bowman_monotonicity} and \citet{kostyshak_2017}.}

\paragraph{Examples of CB tests}

(1) The seminal CB test of \citet{silverman_unimodal} corresponds
to $\mathcal{F}$ as the class of densities, $\mathcal{H}_{0}$ as
the class of densities with less than a specified number of modes,
$\hat{f}$ as a kernel density estimator, and $\mathsf{H}$ as the
set of kernel bandwidths. (2) The test of \citet{bowman_monotonicity}
corresponds to $\mathcal{F}$ as the class of regression functions,
$\mathcal{H}_{0}$ as the collection of monotone regression functions,
and $\hat{f}$ as a non-parametric regression estimator. (3) \citet{harezlak_heckman}
extends $\mathcal{H}_{0}$ to regression functions of an arbitrary
number of modes, and $\hat{f}$ to any estimator with a smoothing
parameter. (4) \citet{kostyshak_2017} extends $\mathcal{H}_{0}$
to quasi-convex and quasi-concave regression functions, and $\hat{f}$
to estimators of generalized additive models.

\subsection{The flatness problem}

For \oldCB{} tests to have good statistical properties, flatness
exclusion is required, which is formally defined as follows:
\begin{defn}[\textbf{flatness exclusion}]
A function $f:\mathcal{X}\mapsto\mathbb{R}$ exhibits \emph{flatness
exclusion} (with respect to $\mathcal{X}\subset\mathbb{R}$) if there
is no interval $[a,b]\subset\mathcal{X}$ such that for $x\in[a,b]$,
${f'(x)=0}$.

\end{defn}
To gain intuition for why the presence of a flat region invalidates
the properties of \oldCB{} tests, consider a simple example in the
context of regression. Suppose that 
\[
y=f(x)+\epsilon,
\]
where $\E(\epsilon|x)=0$ and $f$ is univariate. The \oldCB{} test
is inherently flawed if there exist ${f_{1}\in\mathcal{H}_{0}},{f_{2}\in\mathcal{H}_{0}^{C}}$
that induce indistinguishable test statistics, even in arbitrarily
large samples. This flaw can occur even if $\left\Vert f_{1}-f_{2}\right\Vert _{\infty}$
is large. For an example of how this situation arises, suppose that
$\mathcal{H}_{0}$ is the class of monotone functions and that $f_{1}$
is monotone and contains an interval, say $[a,b]$, over which it
is constant. Suppose that $f_{2}$ is the same as $f_{1}$, except
that it has a dip (violation of monotonicity) in some interval $[c,d]$.
Since the derivative of $f_{1}$ is zero over $[a,b]$, for small
$h$, $\hat{f}_{1}\in\mathcal{H}_{0}^{C}$ with high probability
because even a moderate amount of variation in $\hat{f}_{1}$ leads
to estimates of the derivative on both sides of zero. As $h$ increases,
the variation in $\hat{f}_{1}$ decreases, and eventually $\hat{f}_{1}\in\mathcal{H}_{0}$
with probability increasing toward 1.\footnote{For most semi-parametric estimators, as $h$ tends to infinity, $\hat{f}_{1}$
becomes linear.} Define $h_{stat}^{1}$ to be the CB test statistic associated with
$\hat{f}_{1}$. The identification issue described above occurs if
$\hat{f}_{2}(D,h_{stat}^{1},G)$ is monotone over $[c,d]$: In this
case, the violation of monotonicity of $f_{2}$ over $[c,d]$ was
smoothed away as a result of the flat interval over $[a,b]$, and
$\hat{f}_{2}$ has a similar distribution as $\hat{f}_{1}$.

A similar manifestation of the flatness problem can distort the size
of the CB test and is not driven inherently by the value of the test
statistic, as above, but rather by the bias of the bootstrap. The
validity of the CB bootstrap depends on $\hat{f}(D,\hstatCB,G)$ being
close (in probability) to $f$ under the null, because the bootstrapped
data sets are constructed based on $\hat{f}(D,\hstatCB,G)$. However,
in the presence of flat regions, $\hat{f}(D,\hstatCB,G)$ is  smoother
than $f$. The function $\hat{f}(D,\hstatCB,G)$ has an almost surely
non-zero derivative over the flat interval, and thus the bootstrapped
test statistics are smaller, on average, than the original test statistic,
yielding a low \emph{p}\nobreakdash-value. Even in regions outside
the flat region, the over-smoothed $\hat{f}(D,\hstatCB,G)$ could
be substantially different from $f$, and thus the bootstrapped test
statistic can be further biased. Nominal size is thus smaller than
the actual size.

If $f'(x)$ does not exactly equal zero but is small in absolute value
over an interval, although the CB test is consistent, it can have
poor finite-sample properties. Simulations in Section~\ref{sec:Simulations}
include such a function under the alternative (see $m_{4}$ in Figure~\ref{reg_fns_graph}),
as well as functions that suffer from the flatness problem under the
null (see $m_{1}$ and $flat_{1}$). The next section proposes a
solution to the flatness problem.

\section{Flatness-robust Critical Bandwidth\label{sec:FRCB}}

In this section I introduce flatness-robust critical bandwidth (FRCB),
which transforms \oldCB{} into a test that is robust to the presence
of flat regions in the function of interest.
\begin{defn}[FRCB]
 The FRCB test statistic is defined as
\[
\hstatFRCB\coloneqq h_{stat}(D,\hat{G}^{NF}),
\]
where $h_{stat}$ is the \oldCB{} test statistic, $D$ is the data,
and $\hat{G}^{NF}$ is the grid as constructed in step~\ref{enu:construct_Gnf}
of the following algorithm.
\end{defn}
\begin{lyxalgorithm}
Perform FRCB test at level $\alpha$\label{alg:frcb}~
\begin{enumerate}
\item \label{enu:calc_cbands}Calculate $L_{f'}(x)$ and $U_{f'}(x)$.\footnote{These are constructed at level $\alpha_{n}^{flat}$. For details on
the selection of $\alpha_{n}^{flat}$, see Appendix~\ref{sec:Proofs}.}
\item \label{enu:construct_Gnf}Construct $\hat{G}^{NF}\coloneqq\left\{ x\in G\mid0\notin\left[L_{f'}(x),U_{f'}(x)\right]\right\} $.
\item Perform a \oldCB{} test on $\hat{G}^{NF}$ at level $\alpha$.
\end{enumerate}
\end{lyxalgorithm}
\noindent We now specify sufficient conditions for consistency of
FRCB.
\begin{assumption}
\label{assu:first}\label{assu:CB-consistent-no-flatness}\OldCB{}
is consistent under flatness exclusion.
\end{assumption}
\noindent Theorems in this section demonstrate that FRCB extends consistency
of \oldCB{} to hold without flatness exclusion. Sufficient conditions
for a consistent CB test under flatness exclusion depend on the function
of interest (e.g., density or regression function), as well as the
specific estimation method employed. \noindent For proofs of consistency of CB under flatness exclusion
in specific settings, see \citet{silverman1983} for the case of densities
and \citet{kostyshak_2017} for the case of regression functions.
\begin{assumption}
$f$ is continuously differentiable.
\end{assumption}
\begin{assumption}
\label{assu:supnorm}\label{assu:last}For any $a\in\left(0,1\right)$,
there exists a simultaneous confidence band for $f'$, denoted $\left(U_{f'},L_{f'}\right)$
, such that $\P\left[L_{f'}\le f'\le U_{f'}\right]=a+o(1)$ and $\left\Vert U_{f'}-L_{f'}\right\Vert _{\infty}=o_{p}(1)$.
\end{assumption}
\noindent The following two theorems are the main result of this paper
and show consistency of FRCB without requiring the flatness exclusion
assumption.
\begin{thm}
\label{thm:undernull}Under assumptions \eqref{assu:first}–\eqref{assu:last},
FRCB has asymptotic level $\alpha$.
\end{thm}
\begin{proof}
Proved in Appendix~\ref{sec:proof_undernull}.
\end{proof}
\begin{thm}
\label{thm:underalt}Under assumptions \eqref{assu:first}–\eqref{assu:last},
FRCB has asymptotic power 1.
\end{thm}
\begin{proof}
Proved in Appendix~\ref{sec:proof_underalt}.
\end{proof}
\noindent The convergence rate and asymptotic distribution of CB test
statistics depend on the specific situation (e.g., underlying semi-parametric
estimator of the test), but the following theorem shows that when
$f$ does not have flat regions, the test statistics of FRCB and \oldCB{}
have the same asymptotic distribution.
\begin{thm}
\label{thm:same-dist}Under flatness exclusion, $\hstatFRCB$ has
the same asymptotic distribution as $\hstatCB$.
\end{thm}
\begin{proof}
Proved in Appendix~\ref{subsec:proof_same-dist}.
\end{proof}

\paragraph{Software implementation}

FRCB is implemented in the R\nocite{rcore} package \texttt{frcbstats}
\citep{kostyshak2020_frcbstats}, which makes it easy for practitioners
to apply FRCB to the testing of monotonicity, quasi-convexity, and
quasi-concavity of regression functions. The bootstraps used in \oldCB{}
and in step~\ref{enu:construct_Gnf} of Algorithm~\ref{alg:frcb}
are easily parallelized. Compared with other semi-parametric methods,
the algorithms are fast and the test is practical when working with
millions of observations.

\section{Simulations\label{sec:Simulations}}

This section discusses the results of simulations of a test of regression
monotonicity, and compares rejection rates of the \oldCB{} test with
those of the FRCB test for the regression functions shown in Figure~\ref{reg_fns_graph}.
The regressors are uniformly distributed and the additive error term
is independent and normally distributed with standard deviation 0.25.
The tests use local polynomial regression for the underlying estimator
of $f$ and thin plate regression splines for simultaneous confidence
bands of the derivative.\footnote{This combination of estimators was chosen because local polynomial
regression is well known and thin plate regression splines have nice
properties for estimating the derivative and are the default in R's
\texttt{mgcv} package \citep{wood_mgcv}. For more information, see
\citet{wood_book}. The smoothing parameter for the thin plate regression
is chosen by generalized cross-validation.} 1,000 bootstrap replications are used for the core CB
algorithm and for step \ref{enu:construct_Gnf} of Algorithm~\ref{alg:frcb},
with an $\alpha_{n}^{flat}$ sequence of $n^{-1/2}$. The same estimation
method  is used for the application in the next section.

\begin{figure}[p]

\begin{centering}
\begin{knitrout}
\color{fgcolor}
\includegraphics[width=10cm]{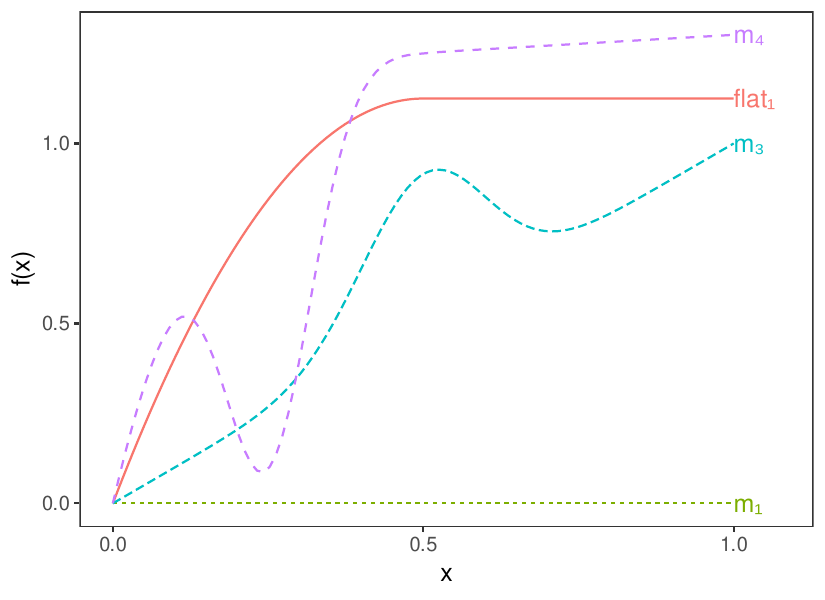} 
\end{knitrout}
\par\end{centering}
\caption{True regression functions used in simulations\label{reg_fns_graph}}

\noindent\begin{minipage}[t]{1\columnwidth}%
\smallskip{}

See Appendix~\ref{sec:sim_functions} for analytical definitions
and origins of functions in this figure.%
\end{minipage}
\end{figure}

\begin{table}[p]

\caption{Rejection Rates of CB and FRCB\label{semiParam_twocov}}

\begin{center}

\begin{tabular}{llllllll}   \hline \hline \multicolumn{1}{l}{} & \multicolumn{7}{c}{$n$} \\ $H_0$ &       50 &      100 &      250 &      500 &  1{,}000 &  2{,}000 & 10{,}000 \\    \hline $M_{CB}(X_{m_{1}})$ & 0.14 & 0.13 & 0.14 & 0.14 & 0.15 & 0.14 & 0.27 \\    $M_{FRCB}(X_{m_{1}})$ & 0.01 & 0.01 & 0.00 & 0.01 & 0.00 & 0.00 & 0.00 \\  \hdashline   $M_{CB}(X_{flat1})$ & 0.41 & 0.50 & 0.62 & 0.73 & 0.84 & 0.94 & 1.00 \\    $M_{FRCB}(X_{flat1})$ & 0.05 & 0.05 & 0.04 & 0.04 & 0.03 & 0.03 & 0.04 \\  \hline   $M_{CB}(X_{m_{3}})$ & 0.14 & 0.16 & 0.36 & 0.69 & 0.95 & 1.00 & 1.00 \\    $M_{FRCB}(X_{m_{3}})$ & 0.06 & 0.16 & 0.43 & 0.78 & 0.98 & 1.00 & 1.00 \\  \hdashline   $M_{CB}(X_{m_{4}})$ & 0.37 & 0.40 & 0.46 & 0.53 & 0.59 & 0.68 & 0.92 \\    $M_{FRCB}(X_{m_{4}})$ & 0.17 & 0.50 & 0.47 & 0.66 & 0.89 & 1.00 & 1.00 \\     \hline \hline \end{tabular}

\end{center}

This table shows the proportion of times each test rejects the null
hypothesis of regression monotonicity at the nominal 5\% level. In
the $H_{0}$ column, the notation ``$M_{T}(Z)$'' means a null hypothesis
of regression monotonicity in $Z$, tested using the test $T$ (CB
or FRCB). Rows above the solid horizontal line are cases in which
the null hypothesis is true. For the shapes of the regression functions,
see Figure~\ref{reg_fns_graph}. Proportions are based on 3,000
simulations.
\end{table}

Table~\ref{semiParam_twocov} shows rejection rates at the nominal
5\% level of CB and FRCB tests for the regression functions graphed
in Figure~\ref{reg_fns_graph}. Rows above the solid horizontal line
are cases in which the null hypothesis is true. In the case of a flat
line ($m_{1}$), the CB test rejects the null hypothesis of monotonicity
more than 10\% of the time (twice the nominal level) across all sample sizes.
A worse scenario for CB is $flat_{1}$, for which the test rejects
in almost half of the samples with a sample size of 100, and the
CB rejection proportion tends to 1 in probability as the sample size
increases. For both $m_{1}$ and $flat_{1}$, the FRCB test rejects
no more than 6\% of the time across all sample sizes.

CB tests are known to be conservative, and FRCB tests inherit this
property. Even with the correction to conservativeness of CB tests
proposed by \citet{hallYork2001_calibration}, which is used in the
simulations, the test is still conservative: Although the test was
run at the nominal 5\% level, the rejection rate of the FRCB test
is closer to 0\% than 5\% for $m_{1}$. More complex calibrations
of the level could be explored but are beyond the scope of this paper.

The rows in the table below the solid line are cases in which the
null hypothesis is false. Both CB and FRCB asymptotically reject,
but for the case of $m_{4}$, the power of the CB test converges more
slowly to 1 than the FRCB test because of the flatness problem that
is triggered by the small derivative of the second half of the function.
For $m_{3}$, in contrast to $m_{4}$, there is no large region with
a derivative close to zero. The rejection rates of CB and FRCB thus
become similar for large sample sizes for this regression function.
The intuition for this similarity is captured in Theorem~\ref{thm:same-dist}.

Two competing effects determine whether the power of FRCB is larger
than CB for small sample sizes. Lack of precision when identifying
flat regions can cause non-flat regions—which might contain evidence
against the null—to be excluded from the $\hat{G}^{NF}$ grid (see
step~\ref{enu:construct_Gnf} of Algorithm~\ref{alg:frcb}) that
is passed to the CB test. This effect explains why CB performs better
than FRCB in the case of $m_{3}$ for sample sizes 50 and 100: The
region just after 0.5 is not estimated precisely enough to determine
whether it is non-flat, and thus the grid on which $\hstatCB$ is
calculated might not contain a strong violation of monotonicity. On
the other hand, a derivative that is relatively small in absolute
value results in low power for the same reason that flatness results
in asymptotic inconsistency. This second effect explains why, for
the case of $m_{4}$, FRCB has higher power than CB for sample sizes
of 100 and larger. Which of the two effects dominates depends on
the shape of the regression function and the precision of the underlying
non-parametric estimator, which is influenced by factors such as the
sample size and the distribution of the error term.

The results of the simulations show how the flatness problem can arise
under both the null and the alternative, and when there are no flat
parts but there are parts with a small first derivative. It is not
surprising that FRCB has lower power than CB when the sample size
is 50: Typically, relying on fewer assumptions has a cost of lower
power. That FRCB has higher power for most sample sizes in the simulation
is a bonus feature, and highlights its usefulness in situations in
which the derivative is not exactly zero but is small enough to trigger
symptoms of the flatness problem.

\section{Application\label{sec:Application}}

\begin{figure}
\begin{knitrout}
\color{fgcolor}
\includegraphics[width=\maxwidth]{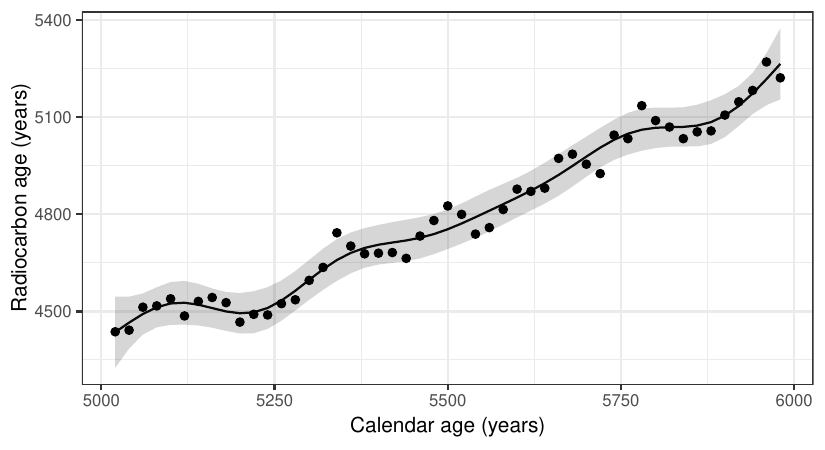} 
\end{knitrout}

\caption{Raw data with non-parametric point and interval estimates of the conditional
mean function\label{fig:unrestricted}}

\smallskip{}

{\footnotesize This figure shows a non-parametric fit and corresponding
simultaneous confidence band at the 95\% confidence level. The estimates
are from thin plate regression splines with the smoothing parameter
chosen by generalized cross-validation. The simultaneous confidence
band is constructed from an empirical bootstrap with 1,000
replications; the critical value is the bootstrapped 0.95-quantile
of the maximal t-statistic.}
\end{figure}

We examine a subset of radiocarbon data from Irish oak trees, published
by \citet{pearson1993high} and previously analyzed in the context
of regression monotonicity using CB by \citet{bowman_monotonicity}.
The dataset consists of the radiocarbon age and the calendar age.\footnote{A measure of precision of the radiocarbon dates is also included in
the dataset, but for simplicity the variable is ignored in this paper
since \citet{bowman_monotonicity} found that incorporating the measure
did not substantively change their results.} Calibration of this relationship is important, as the radiocarbon
age is often used to predict the calendar age when the calendar age
is unknown. The raw data are shown in Figure~\ref{fig:unrestricted},
along with a non-parametric fit and corresponding simultaneous confidence
band at the 95\% confidence level. The relationship is known to be
non-monotone,\footnote{In the radiocarbon dating literature, these fluctuations are sometimes
referred to as ``wiggles'' or ``de Vries effects.'' For more information,
see \citet[pp.\ 53--54]{taylor2014radiocarbon}.} and thus is an interesting case study for whether the non-monotonicity
can be picked up with a small sample size. A visual inspection of
the simultaneous confidence band does not suggest strong evidence
against monotonicity. However, visual inspection of the data points
themselves does hint at non-monotone fluctuations in the underlying
regression function.

A \oldCB{} test, using the same non-parametric estimation method
as the simulations, results in a \emph{p}-value of 0.002.
Fitted values corresponding to three selected smoothing parameters,
including the test statistic value, are shown in Figure~\ref{fig:paths-cb}.
Consider the fitted curve corresponding to a smoothing parameter of
37.8, which is
considerably smaller than the test statistic value of 56.3.
The 37.8 fit
is monotone, except for the section with calendar ages less than 5,250
years. This shows that the flatness problem could be binding, i.e.,
both monotone and non-monotone regression functions exist that would
yield the same distribution of the test statistic. Another indication
that flatness could be a problem for this application is that the
three fitted curves have different signs of the derivative over considerable
portions of the segment before 5,250. In summary, the
small CB \emph{p}-value could be driven by the 12
data points with values less than 5,250, and it is not
clear whether those data points do indeed carry enough information
to suggest such a strong rejection of monotonicity.

Given the concerns regarding the flatness problem described above,
it is interesting to consider using FRCB to examine whether there
exists strong evidence of non-monotonicity when regions with a small
derivative (e.g., the region before 5,250) are detected
and excluded.  The FRCB test statistic is 38.3,
which is smaller than the CB statistic (56.3).
The FRCB test statistic is usually smaller than the CB test statistic
because the set of FRCB inequality constraints is a subset of the
CB inequality constraints. In this particular application the difference
is large, which is not surprising given the flatness concerns. The
fitted value curves corresponding to smoothing parameters from the
FRCB test are shown in Figure~\ref{fig:paths-frcb}. The gaps between
the line segments are a result of the filtering of the CB grid from
step~\ref{enu:construct_Gnf} of Algorithm~\ref{alg:frcb}, and
are useful for understanding where the derivative can be interpreted
away from zero with precision. The FRCB \emph{p}-value is 0.013,
and thus the non-monotone fluctuations in the sample provide some
evidence of non-monotonicity in the population.

\begin{figure}
\subfloat[\OldCB{}\label{fig:paths-cb}]%
{
\begin{knitrout}
\color{fgcolor}
\includegraphics[width=.5\linewidth]{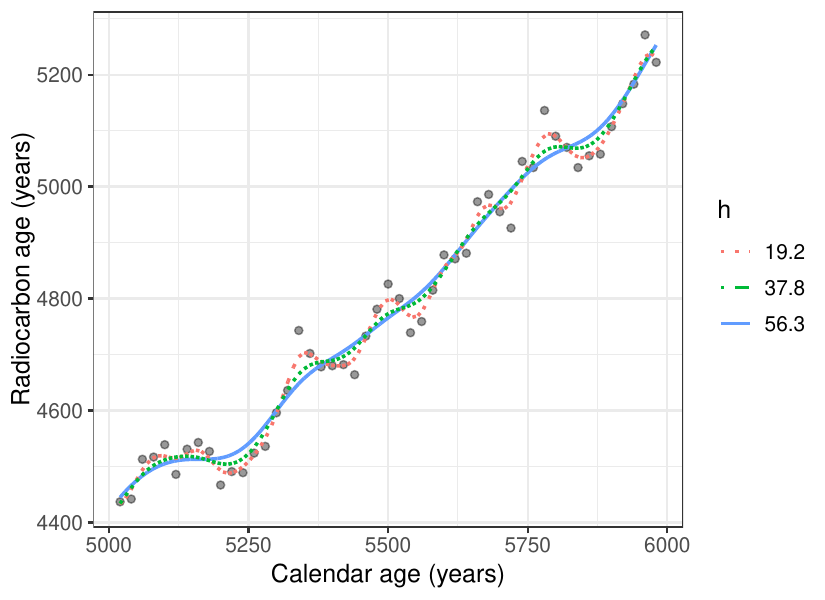} 
\end{knitrout}

}%
\subfloat[FRCB\label{fig:paths-frcb}]%
{
\begin{knitrout}
\color{fgcolor}
\includegraphics[width=.5\linewidth]{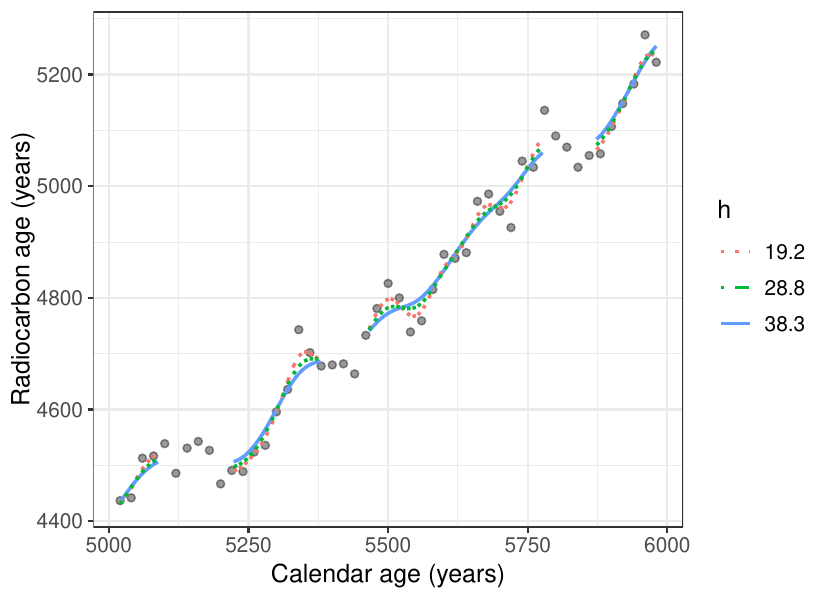} 
\end{knitrout}

}

\caption{Fits for various smoothing parameters, CB and FRCB}
\end{figure}

\section{Conclusion\label{sec:Conclusion}}

This paper provides an asymptotic solution to the fundamental problem
that in some situations, even with a large sample size, CB does not
provide correct inference.  Further, simulations show that FRCB performs
better than CB in some finite-sample situations, even when no regions
are exactly flat. The application to radiocarbon data provides insight
on the flatness problem in an applied situation, and shows how FRCB
can be used to provide inference that is robust to flat regions.

By removing an assumption that in practice might be violated, FRCB
offers a robust improvement, and, in turn, opportunities to apply
critical bandwidth methods to the large class of potential applications.

\pagebreak{}

\printbibliography[heading=bibintoc]

\pagebreak{}

\appendix

\section{Proofs\label{sec:Proofs}}

\subsection{Preliminaries}

The following definitions are used to simplify the proofs below. $L_{f'}$
and $U_{f'}$ are as defined in assumption~\eqref{assu:supnorm},
and $\hat{G}^{NF}$ is as constructed in Algorithm~\ref{alg:frcb}.
Define $R^{CB}$ to equal 1 if the \oldCB{} test rejects, and 0 otherwise.
Similarly, define $R^{FRCB}$ to equal 1 if the FRCB test rejects,
and 0 otherwise. Let $G^{F}$ and $G^{NF}$ partition $G$ such that
$G^{F}=\left\{ x\in G\mid f'(x)=0\right\} $, $G^{NF}=\left\{ x\in G\mid f'(x)\ne0\right\} $.
Construct $\alpha_{n}^{flat}$ to be any sequence that satisfies both
$\alpha_{n}^{flat}\rightarrow0$ and assumption~\eqref{assu:supnorm};
that is, if we define $d_{n}(\alpha_{n}^{flat})=\left\Vert U_{f'}-L_{f'}\right\Vert _{\infty}$
to be the sequence of stochastic sup-norm lengths of the confidence
band, we require that both $\alpha_{n}^{flat}\rightarrow0$ and $d_{n}(\alpha_{n}^{flat})\overset{p}{\rightarrow}0$.\footnote{Intuitively, we require $\alpha_{n}^{flat}$ to converge to 0, but
at a rate that is sufficiently slow that enough precision is still
gained. The rate will depend on the properties of the particular CB
test and simultaneous confidence band.

}

\subsection{\label{sec:proof_undernull}Proof of Theorem~\ref{thm:undernull}}
\begin{thm*}
Under assumptions \eqref{assu:first}–\eqref{assu:last}, FRCB has
asymptotic level $\alpha$.
\end{thm*}
\begin{proof}
Under the null hypothesis,
\begin{align*}
\P\left[R^{FRCB}=1\right]=\,\,\, & \P\left[R^{CB}=1\mid\hat{G}^{NF}\cap G^{F}=\emptyset\right]\P\left[\hat{G}^{NF}\cap G^{F}=\emptyset\right]+\\
 & +\P\left[R^{CB}=1\mid\hat{G}^{NF}\cap G^{F}\ne\emptyset\right]\P\left[\hat{G}^{NF}\cap G^{F}\ne\emptyset\right]\\
\le\,\,\, & \alpha\left(1-\alpha_{n}^{flat}\right)+\P\left[R^{CB}=1\mid\hat{G}^{NF}\cap G^{F}\ne\emptyset\right]\alpha_{n}^{flat}\\
=\,\,\, & \alpha+\alpha_{n}^{flat}\left\{ -1+\P\left[R^{CB}=1\mid\hat{G}^{NF}\cap G^{F}\ne\emptyset\right]\right\} \\
\rightarrow\,\,\, & \alpha,
\end{align*}
where the convergence is a result of $\alpha_{n}^{flat}\rightarrow0$.
\end{proof}

\subsection{\label{sec:proof_underalt}Proof of Theorem~\ref{thm:underalt}}
\begin{thm*}
Under assumptions \eqref{assu:first}–\eqref{assu:last}, FRCB has
asymptotic power 1.
\end{thm*}
\begin{proof}
Define $A$ to be the event $\left\{ G^{NF}=\hat{G}^{NF}\right\} =\left\{ \hat{G}^{NF}\cap G^{F}=\emptyset,\,G^{NF}\subseteq\hat{G}^{NF}\right\} $.
Then,
\begin{align*}
\P\left[A^{C}\right]\le\,\,\, & \P\left[\hat{G}^{NF}\cap G^{F}\ne\emptyset\right]+\P\left[G^{NF}\not\subseteq\hat{G}^{NF}\right]\\
=\,\,\, & \left(\alpha_{n}^{flat}+o(1)\right)+o(1)\\
\rightarrow\,\,\, & 0,
\end{align*}
where the $o(1)$ terms come from assumption~\eqref{assu:supnorm}.
Then, under the alternative,
\begin{align*}
\P\left[R^{FRCB}=1\right]=\,\,\, & \P\left[R^{CB}=1\mid A\right]\P\left[A\right]+\P\left[R^{CB}=1\mid A^{C}\right]\P\left[A^{C}\right]\\
\rightarrow\,\,\, & 1,
\end{align*}
because $\P\left[R^{CB}=1\mid A\right]\rightarrow1$ by assumption~\eqref{assu:CB-consistent-no-flatness},
and $\P\left[A^{C}\right]\rightarrow0$ as shown above.
\end{proof}

\subsection{\label{subsec:proof_same-dist}Proof of Theorem~\ref{thm:same-dist}}
\begin{thm*}
Under flatness exclusion, $\hstatFRCB$ has the same asymptotic distribution
as $\hstatCB$.
\end{thm*}
\global\long\def\hlim{h_{\infty}}%

\begin{proof}
Suppose that $a_{n}\left(\hstatCB-\hlim\right)\convergesInDistribution W$
for some sequence $a_{n}$, constant $\hlim$, and distribution r.v.\ $W$.
In most applications of \oldCB{}, $\hlim$ is 0 under the null hypothesis.

Define 
\begin{align*}
\mathcal{C}(x) & \coloneqq\left\{ y\in\mathbb{R}\mid L_{f'}(x)\le y\le U_{f'}(x)\right\} ,\\
I & \coloneqq\bigcup_{g\in G}\mathcal{C}(g).
\end{align*}
Consider that under flatness exclusion,
\[
\P\left[\hat{G}^{NF}=G\right]=\P\left[0\notin I\right].
\]
Define $B$ to be the event $\left\{ f'\left(g\right)\in\mathcal{C}(g)\text{ for all }g\in G\right\} $.
$\P\left[B\right]=1-\left(\alpha_{n}^{flat}+o(1)\right)$ by step~\eqref{enu:calc_cbands}
of Algorithm~\ref{alg:frcb} and assumption~\eqref{assu:supnorm}.
Define $a=\min_{g\in G}\left|f'(g)\right|$. Under flatness exclusion,
$a>0$. Define $Q(x)=\left|U_{f'}(x)-f'(x)\right|\vee\left|L_{f'}(x)-f'(x)\right|$.
Then,

\[
\P\left[0\in I|B\right]\le\P\left[\max_{g\in G}Q(g)\ge a\right],
\]
which converges to 0 by assumption~\eqref{assu:supnorm}. It follows
that 
\begin{align*}
\P\left[0\in I\right] & =\P\left[0\in I|B\right]\P\left[B\right]+\P\left[0\in I|B^{C}\right]\left(1-\P\left[B\right]\right)\\
 & =\P\left[0\in I|B\right]\left[1-\left(\alpha_{n}^{flat}+o(1)\right)\right]+\P\left[0\in I|B^{C}\right]\left(\alpha_{n}^{flat}+o(1)\right),
\end{align*}
 where $\P\left[0\in I|B\right]\rightarrow0$ by above, and $\alpha_{n}^{flat}\rightarrow0$
by construction. Then, $\P\left[0\in I\right]\rightarrow0$, so $\P\left(\hat{G}^{NF}=G\right)\rightarrow1$,
and thus
\[
\P\left[\hstatFRCB=\hstatCB\right]=\P\left[\hstatCB\left(D,\hat{G}^{NF}\right)=\hstatCB\left(D,G\right)\right]\rightarrow1,
\]
i.e., $\hstatFRCB\convergesInProbability\hlim$. It also follows
that
\begin{align*}
 & \P\left[a_{n}\left(\hstatFRCB-\hlim\right)=a_{n}\left(\hstatCB-\hlim\right)\right]\\
 & =\P\left[a_{n}\left\{ \hstatCB\left(D,\hat{G}^{NF}\right)-\hlim\right\} =a_{n}\left\{ \hstatCB\left(D,G\right)-\hlim\right\} \right]\rightarrow1,
\end{align*}
and thus $a_{n}\left(\hstatFRCB-\hlim\right)\convergesInDistribution W$
by, e.g., Theorem\ 2.7 of \citet{vandervaart1998_book}. 
\end{proof}

\section{\label{sec:sim_functions}Simulation Functions}

This section provides the analytical definitions and origins of the
functions graphed in Figure~\ref{reg_fns_graph}.  The definitions
are as follows:
\begin{align*}
m_{1}\left(x\right) & =0\\
m_{3}\left(x\right) & =x+0.415e^{-50\left(x-0.5\right)^{2}}\\
m_{4}\left(x\right) & =\begin{cases}
10\left(x-0.5\right)^{3}-e^{-100\left(x-0.25\right)^{2}} & \mbox{if }x<0.5,\\
0.1\left(x-0.5\right)-e^{-100\left(x-0.25\right)^{2}} & \mbox{otherwise}
\end{cases}\\
flat_{1}\left(x\right) & =\begin{cases}
0.5x(1-x)\cdot9 & \mbox{if }x<0.5,\\
0.5^{3}\cdot9 & \mbox{otherwise}
\end{cases}
\end{align*}
The $m_{n}$ functions were explored by \citet{gsv_mono}, and some
were also studied in other papers that tested for monotonicity of
regression functions: At least one of $m_{3}$ and $m_{4}$ was
used by \citet{bowman_monotonicity}, \citet{hallHeckman_mono_dip},
and \citet{kostyshak_2017}. Functions similar to $flat_{1}$ were
explored by \citet{simonsohn2017}.
\end{document}